%% file: main.tex
\documentclass{llncs}

\usepackage[utf8]{inputenc}

\input{header_common}

\input{header}

\title{Different Maps for Different Uses}
\subtitle{A Program Transformation for \\ Intermediate Verification Languages}

\author{
  Daniel Dietsch \and
  Matthias Heizmann \and
  Jochen Hoenicke \and \\
  Alexander Nutz \and
  Andreas Podelski
 }
\institute{University of Freiburg}

\begin{document}

\maketitle


 \begin{abstract}
  \input{abstract}

 \end{abstract}

 \section{Introduction}
  \input{introduction}

 \section{Example}
  \label{chap_heapsep_sec_example}
  \input{example}

 \section{Preliminaries}
  \input{preliminaries}

 \section{Dependency Analysis}
  \label{chap_heapsep_sec_dataflow}
  \input{dataflow}

 \section{Computing Dependencies}
  \label{chap_heapsep_sec_preprocessing}
  \input{preprocessing}

 \section{Program Transformation}
  \label{chap_heapsep_sec_transformation}
  \input{transformation}

 \section{Implementation in Ultimate}
  \input{extensions}

\section{Experiments on a Scalable Benchmark Suite}
  \input{evaluation}

 \section{Related Work}
  \input{relatedwork}

 \section{Discussion}
  \label{chap_heapsep_sec_discussion}
  \input{discussion}
 \section{Conclusion}
  \input{conclusion}
 \bibliographystyle{abbrv}
 \bibliography{main}

\end{document}

%% file: header_common.tex
\usepackage[utf8]{inputenc}

\usepackage[hidelinks, breaklinks=true]{hyperref}

\usepackage{bbding}

\usepackage{etex}
\usepackage[table]{xcolor} 
\usepackage[utf8]{inputenc}

\usepackage{etoolbox} 

\usepackage{float}
\usepackage{xspace}
\usepackage{calc}
\usepackage{caption}
\usepackage{subcaption}
\usepackage{enumitem}
\usepackage{booktabs}
\usepackage{tabu}
\usepackage{multirow}
\usepackage{url}
\usepackage{graphicx} 
\usepackage{longtable} 
\usepackage{dcolumn}

\usepackage{wrapfig}

\usepackage{amsmath, amssymb}
\usepackage{extarrows}
\usepackage{mathpartir}
\usepackage{stmaryrd}

\usepackage{listings}
\input{colors}

\usepackage{courier}

\usepackage[colorinlistoftodos]{todonotes}
\usepackage{hyperref} 

\usepackage{tikz}
\usetikzlibrary{fit, calc, tikzmark, positioning}

\usetikzlibrary{decorations.pathreplacing}

\usepackage{pgfplots} 

\hypersetup{colorlinks  
  ,pdfpagemode=UseNone
  ,linkcolor=black  
  ,filecolor=black 
  ,urlcolor=black
  ,citecolor=black
} 

\lstdefinelanguage{Boogie}{
  morekeywords={assert, assume, while, if, then, else, procedure, 
    implementation, return, var, call, ensures, requires, havoc, modifies, type},
  sensitive=false,
  morecomment=[l]{//},
  morestring=[b]",
  escapeinside={/*}{*/},
}
\newcommand{\define}{\ensuremath{\stackrel{\mathit{def}}{=}\ }}

\newcommand{\locations}{\ensuremath{\mathsf{Loc}}\xspace}
\newcommand{\constarray}[1]{\ensuremath{(\mathsf{const}\ #1)}}
\newcommand{\subst}[2]{\ensuremath{[#1\!\mapsto\!#2]}}

\newcommand{\true}{\ensuremath{\texttt{true}}}
\newcommand{\false}{\ensuremath{\texttt{false}}}

\newcommand{\loc}{\ensuremath{\mathit{\ell}}\xspace} 

\newcommand{\initloc}{\ensuremath{\mathit{\loc_{0}}}\xspace}

\newcommand{\preimage}[2]{\ensuremath{#1^{-1}[#2]}}

\renewcommand{\sharp}{\ensuremath{^\#}}

\newcommand{\lit}{\ensuremath{\mathsf{lit}}}

\newcommand{\post}[2]{\mathsf{post}(#1, #2)}
\newcommand{\postOp}{\mathsf{post}}

\newcommand{\statements}{\ensuremath{\Sigma}}
\newcommand{\allstatements}{\ensuremath{\mathsf{Statements}}}
\newcommand{\bnf}{\,::=\,}
\newcommand{\states}{\ensuremath{\mathsf{States}}}

\newcommand{\allcommands}{\ensuremath{\mathsf{Commands}}}

\newcommand{\allexpressions}{\ensuremath{\mathsf{Expressions}}\xspace}
\newcommand{\allexecutions}{\ensuremath{\mathsf{Executions}}}
\newcommand{\allstates}{\ensuremath{\states}}

\newcommand{\allvariables}{\ensuremath{\mathsf{Variables}}}

\newcommand{\unintsort}{\ensuremath{\mathit{Sort}}\xspace}
\newcommand{\mapsort}{\ensuremath{\unintsort\! \to \unintsort}\xspace}

\newcommand{\addspace}[1]{\ldotp#1}
\newcommand{\execution}[2]{#1 \forcsvlist\addspace{#2} }

\newcommand{\suffixbase}{\ensuremath{\mathsf{base}}}
\newcommand{\suffixmap}{\ensuremath{\mathsf{map}}}
\newcommand{\suffixbool}{\ensuremath{\mathsf{bool}}}

\newcommand{\baseexp}{\ensuremath{\mathit{e}_\suffixbase}}

\newcommand{\mapexp}{\ensuremath{\mathit{e}_\suffixmap}}

\newcommand{\boolexp}{\ensuremath{\mathit{e}_\suffixbool}}

\newcommand{\ltldiamond}{\ensuremath{\lozenge\,}}

\newcommand{\variables}{\ensuremath{\mathit{Var}}\xspace}

\newcommand{\mapvars}{\ensuremath{\variables_\suffixmap}\xspace}

\newcommand{\progand}{\ensuremath{\,\texttt{\&\&}\,}}
\newcommand{\progor}{\ensuremath{\,\texttt{||}\,}}
\newcommand{\prognot}{\ensuremath{\texttt{!}}}

\newcommand{\reach}{\ensuremath{\textit{Reach}}\xspace}

\newcommand{\ultimate}{Ultimate\xspace}

%% file: colors.tex
\definecolor{a1}		{RGB}{215,25,28} 
\definecolor{a2}		{RGB}{150,150,20} 
\definecolor{a3}		{RGB}{44,123,182} 
\definecolor{a4}		{RGB}{20,255,191} 
\definecolor{a5}		{RGB}{171,217,233} 

\definecolor{dataflow}		{RGB}{20,20,233} 

\definecolor{orange}   		{RGB}{255,128,0}
\definecolor{darkred}  		{RGB}{128,0,0}
\definecolor{darkgreen}		{RGB}{0,128,0}
\definecolor{dirtgreen}		{RGB}{180,210,180}
\definecolor{mixgreen}		{RGB}{34,139,34}
\definecolor{darkblue} 		{RGB}{0,0,128}
\definecolor{darkpurple}	{RGB}{160,32,240}
\definecolor{lightpurple}	{RGB}{180,180,210}
\definecolor{bl1}   		{RGB}{204,204,255}
\definecolor{bl2}   		{RGB}{128,128,255}
\definecolor{bl3}   		{RGB}{140,160,255}

\definecolor{gr1}		{RGB}{250, 250, 250}
\definecolor{gr2}		{RGB}{229, 229, 229}
\definecolor{gr3}		{RGB}{212, 212, 212}
\definecolor{gr4}		{RGB}{204, 204, 204}

\definecolor{g1}		{RGB}{215,25,28} 
\definecolor{g2}		{RGB}{253,174,97} 
\definecolor{g3}		{RGB}{255,255,191} 
\definecolor{g4}		{RGB}{171,217,233} 
\definecolor{g5}		{RGB}{44,123,182} 

\definecolor{s1}		{RGB}{228,26,28}
\definecolor{s2}		{RGB}{55,126,184}
\definecolor{s3}		{RGB}{77,175,74}
\definecolor{s4}		{RGB}{152,78,163}
\definecolor{s5}		{RGB}{255,127,0}
\definecolor{s6}		{RGB}{255,255,51}
\definecolor{s7}		{RGB}{166,86,40}
\definecolor{s8}		{RGB}{247,129,191}
\definecolor{s9}		{RGB}{153,153,153}

\colorlet{codegreen}		{mixgreen}
\colorlet{codepurple}		{darkpurple}
\colorlet{codehighlight}	{gr2}

\colorlet{stmtcolor}		{gr2} 
\colorlet{prg}				{gr2} 
\colorlet{ltl}				{g4} 
\colorlet{notIC}			{bl3} 
\colorlet{stateass}			{g2}

\colorlet{outlineblue}		{g5}
\colorlet{fillblue}			{g4}
\colorlet{darkback}			{gr2}
\colorlet{lightback}		{gr1}

\colorlet{source}			{gr1}
\colorlet{controller}		{gr1}
\colorlet{thirdparty}		{gr1} 
\colorlet{core}				{gr1}
\colorlet{libs}				{gr1}
\colorlet{source}			{gr1}
\colorlet{analyzer}			{gr1}
\colorlet{generator}		{gr1}
\colorlet{output}			{gr1}

%% file: header.tex
\newcommand{\writeblocks}{\ensuremath{\mathcal{W}}\xspace}
\newcommand{\writeblock}{\ensuremath{W}\xspace}
\newcommand{\somewrite}{\ensuremath{\sigma_{\sf wr}}\xspace}
\newcommand{\someread}{\ensuremath{\sigma_{\sf rd}}\xspace}

\newcommand{\stmtransformer}{\ensuremath{\tau}\xspace}
\newcommand{\rdtransformer}{\ensuremath{\tau_\lastwrites}\xspace}

\newcommand{\lastwrite}{\ensuremath{\mathsf{lw}}\xspace}
\newcommand{\lastwrites}{\ensuremath{\mathsf{LstWr}}\xspace}
\newcommand{\lastwritessharp}{\ensuremath{\lastwrites^{\sharp}\!}\xspace}
\newcommand{\nonwrite}{\ensuremath{\bot}\xspace}

\newcommand{\stmsrcloc}{\ensuremath{\mathit{src}}}

\newcommand{\bisim}{\ensuremath{\sim}\xspace}

\newcommand{\progpre}{\ensuremath{P_\lastwrites}\xspace}

\newsavebox\heapsepExampleLBox
\newsavebox\heapsepExampleRBox

\newcommand{\stm}[1]{\ensuremath{\texttt{#1}}\xspace}
\newcommand{\stmem}[1]{\ensuremath{\emph{\texttt{#1}}}\xspace}
\newcommand{\stmselect}[2]{\stm{#1[#2]}}
\newcommand{\stmstore}[3]{\stm{#1[#2:=#3]}}
\newcommand{\stmconstarray}[1]{\ensuremath{\stm{(const $#1$)}}}

\newcommand{\stmedge}[3]{\ensuremath{(#1, \, #2 \, , #3)}}

\newcommand{\writestatements}{\ensuremath{\Sigma_\mathsf{wr}}\xspace}
\newcommand{\readstatements}{\ensuremath{\Sigma_\mathsf{rd}}\xspace}

\newcommand{\eval}[2]{\ensuremath{#1[\![#2]\!]}}

\newcommand{\automizer}{{\sf Automizer}\xspace}



%% file: abstract.tex
In theorem prover or SMT solver based verification, the program to be verified
is often given in an intermediate verification language such as Boogie, Why, or
CHC.  This setting raises new challenges.  We investigate a preprocessing step
which takes the similar role that alias analysis plays in verification, except
that now, a (mathematical) map is used to model the memory or a data object of
type {\tt array}.  We present a program transformation that takes a program $P$
to an equivalent program $P'$ such that, by verifying $P'$ instead of $P$, we
can reduce the burden of the exponential explosion in the number of case splits.
Here, the case splits are according to whether two statements  using the same
map variable are independent or not; if they are independent, we might as well
employ two different map variables and thus remove the need for a case split
(this is the idea behind the program transformation).  We have implemented the
program transformation and show that,  in an ideal case, we can
avoid the exponential explosion.

%% file: introduction.tex
In theorem prover or SMT solver based verification, the
program to be verified is often given in an \emph{intermediate verification
language} (such as Boogie~\cite{this-is-boogie-2-2}, Why~\cite{Filliatre03}, or CHC~\cite{DBLP:conf/pldi/GrebenshchikovLPR12}). This setting
is useful in many aspects but it raises its proper challenges; see,
e.g.,~\cite{DBLP:conf/paste/BarnettL05} for the investigation of axiomatic
semantics. Here,
we investigate a novel problem that arises in this setting
where a (mathematical) map is used to model the memory or a data
object of type \texttt{array}.  The problem is to transform a program
into an equivalent program such that statements with
independent uses of a given map variable become statements with
different map variables.  In a way, we lift the \emph{alias problem}
from programming languages to intermediate verification languages.  We
will next explain the problem and the new challenge that it raises.
The explanation is subtle and will need a large chunk of the introduction.



The idea behind an intermediate verification language is the one
of a \emph{lingua franca} for verification. Once a C or Java program
has been translated to intermediate code, we are no longer bothered
with the intricacies and ambiguities of definitions of programming
language semantics. There are no {hidden assumptions}
(such as, e.g.,  the absence of undefined behavior); the program is
taken \emph{as is}, i.e., all assumptions appear in the program text
(e.g., in \texttt{ensure} statements). This is one reason why it has been advocated to
present benchmarks in an intermediate programming language for
software verification competitions; see, e.g.,~\cite{chccomp}.  Note that there
are several scenarios where the program in the intermediate
verification language comes without a corresponding program in a
programming language.  For example, it may have been constructed by a
specific module of the
verification method; see, e.g., the construction of \emph{path
  programs} in~\cite{DBLP:conf/pldi/BeyerHMR07} and~\cite{DBLP:conf/sas/GreitschusDP17}.


The data manipulated by a program in an intermediate verification
language are mathematical objects (in the same domains and logical theories that
underly the theorem prover or SMT solver used for the verification). In particular,
an object of type
\texttt{array} in the intermediate verification language
is, in fact, a {map} in the mathematical sense (i.e., it is
manipulated like a mathematical map). 

The importance of
maps in intermediate verification languages is inherited directly from the
importance of arrays in programming languages. The importance is
amplified by the fact that in verification it is often convenient to
view the memory (or, the \emph{heap})  as a special case of an
array.

What is also inherited is, unfortunately, a
notorious practical issue in program verification:  the
need of case splits according to whether two statements with a
write resp.\ read access to a given array (or, to the memory) refer to the
same position, or not.   A well-known consequence of such case splits
is that they can lead to the exponential explosion of the size of the verification
condition.  If before we had
the exponential explosion   in
the number of statements in the program that access a given array (or
the memory), we
have now have the exponential explosion in the number of statements
that use a given map.  We thus need to address an analogous issue in the context of
intermediate
verification languages.

The standard solution to address the notorious practical issue is a
preprocessing step with an \emph{alias analysis}. Roughly speaking,
the alias analysis can help to infer which case splits are redundant.
In some cases, the alias analysis can thus alleviate the burden of the
exponential explosion in the number of case splits.

Unfortunately, an alias analysis for programming languages cannot
readily be transferred to a solution for intermediate verification languages.
The new challenge stems from the fact that we assume that a program in the intermediate verification language will encode
every assumption in the program text; i.e.,  we are not allowed
to use any assumption that does not appear in the program text.  

We give an example to illustrate this point.  
The example program is depicted in Figure~\ref{chap_heapsep_fig_example_ma} in 
Section~\ref{chap_heapsep_sec_discussion}.
We here use the
map-valued variable \texttt{mem} to model the memory and the procedure
\texttt{malloc} to model allocation (which we can specify together
with \texttt{ensure} statements that encode our assumptions about
allocation).  A statement that uses the map \texttt{mem} at position
\texttt{p} intuitively models the access of memory (by a write or by a read) via the pointer variable \texttt{p}.
We take a program that contains two statements which use \texttt{mem} at position
\texttt{p}  and position \texttt{q}, respectively.  We would like to infer that the
uses of \texttt{mem}  in the two statements are independent.  The term \emph{independent} here
means that the value of \texttt{p} in the execution of the one
statement is different from the value of \texttt{q} in the execution of the other
statement (in every execution of the program).
In our setting, we are
not allowed to use any hidden assumption (such as the absence of
undefined behavior). For example, we
are not allowed to assume that there is no execution in which the two statements are  executed when the value of \texttt{p} and
\texttt{q} is \texttt{null}. Thus, we are not
allowed to conclude that the two uses of \texttt{mem} in the two
statements are independent 
even if we can infer that \texttt{p} has not been assigned to \texttt{q},
and vice versa.  This would not be \emph{sound}.

Note that in the context of programming languages, where it is common
to use the assumption of the absence of undefined behavior, it
would be considered sound to conclude ``\texttt{p} and \texttt{q} do
not alias'' if the analysis can infer from the property that there is no execution in the
program that assigns \texttt{p} to \texttt{q}, and vice versa.   In
this sense, the
hidden assumption is the basis for the existence of 
very efficient (and effective) alias analyses.  A static analysis can infer the property by
checking a strong sufficient condition for the property (e.g., that the corresponding
statements simply do not occur in the program).

\subsubsection*{Contributions}

The overall contribution of this paper is to investigate the
theoretical foundations and a preliminary solution for a novel
research question which may be relevant for the practical potential of
intermediate verification languages. 

The question concerns a preprocessing step for intermediate
verification languages which takes the similar role that alias
analysis plays in the verification for programming languages.  Since
it is convenient to implement an optimization as a program
transformation (in particular for intermediate code), we consider a program
transformation that takes a program $P$ to an equivalent program $P'$ such that, by verifying $P'$ instead of $P$, we
can reduce the burden of the exponential explosion in the number of
case splits.  Here, the case splits are according to whether two statements in $P$
using the same map variable are independent or not; if they are
independent, we might as well employ different map variables  and thus
remove the need for a case split (this is
the idea behind the program transformation).   The question is:
Does there exist such a program transformation, and can it be made scalable?

In this paper, we present such a program transformation, together with
its implementation which we use to show that, in the best
case, we can avoid the exponential explosion altogether.

The program transformation is based on a static analysis that conservatively
infers which statements using a giving map variable are independent.
The overall goal of the
 analysis is to
infer a \emph{grouping} of statements such that we can introduce a
different map variable for each group of statements (the statements
within each group use the same 
map variable).


\goodbreak

Our technical contributions are as follows. 
\begin{itemize}
\item We formally introduce the independence property which enables the desired
  program transformation (in the context for of the intermediate
  programming language). 

\item We present a static analysis that conservatively
infers which statements using a giving map variable are independent.
We define an instrumentation of a program
  with auxiliary variables such that an existing static analysis
  can infer the independence property. 
\item We define a program
  transformation that takes as input a program and the inferred
  independence property and returns a new program. In the new program
   statements use different map variables according to the
  inferred indepence property. 
\item We prove that the program
  transformation is sound, i.e., the new program is bisimulation
  equivalent to the input program.
\item We have implemented the program transformation into a toolchain
  for automatic verification.  A preliminary
  experimentation shows that the program transformation
  can be effective, at least in principle.   On a benchmark suite which is specifically tailored
  to condensate the case split explosion problem, the toolchain with the
  program transformation scales very well in the size of the program
  (whereas the toolchain without the program transformation
  quickly falls into the case split explosion problem and runs out of time
  or space).
\end{itemize}

%% file: example.tex
The left hand side of Figure \ref{fig_heapsep_example_original_prog} shows an
example program given in the Boogie~\cite{this-is-boogie-2-2}
verification language.
While the program models a program in the C programming language we want to
stress that our technique cannot rely on any metainformation specific to C, like
the meaning of the \stm{malloc} procedure, or the absence map reads on
uninitialized cells.
Map semantics in Boogie follow McCarthy's theory of 
arrays~\cite{DBLP:conf/ifip/McCarthy62}, which is also used in SMT solvers.


%

The example program is artificial.  Its purpose is to necessitate a large number
of non-interference checks in a program of minimal size. So the main obstacle to
verifcation is the necessity of proving non-interference between the map
updates.

\input{fig/exampleProg}

In the example, dynamically allocated memory is modeled by the two map variables
$\stm{mem}$ and $\stm{valid}$.
The map $\stm{mem}$ stores the contents of the memory.  The map $\stm{valid}$
stores which memory cells are allocated.  C's malloc function is modeled by the
procedure $\stm{malloc}$, which returns a memory location that is not
currently in use. (For simplicity we assume that all memory blocks are of size
1.)

The procedure \stm{main} starts by allocating two pointers and storing them to
variables \stm{p} and \stm{q}.  The contents of both memory locations \stm{p}
and \stm{q} are initialized to 0.  Then, the value at location \stm{p} is
incremented nondeterministically often, and the value at location \stm{q} is
decremented nondeterministically often. The \stm{assert} statements express
that, at the end of the program the values in memory at \stm{p} and \stm{q}
contain a non-negative or a non-positive value respectively.

As an intermediate goal to correctness, a solver must prove that the operations
on memory cells \stm{p} and \stm{q} do not interfere.  A typical 
CEGAR-based, or bounded model 
checking-based, solver will need to do this for every spurious
counterexample.

Our technique provides a preprocessing such that the solver can instead prove
correctness of the transformed program on the right hand side of Figure
\ref{fig_heapsep_example_transformed_prog}.
In the transformed example, the map \stm{mem} has been
replaced by two maps \stm{mem\_1} and \stm{mem\_2}.
Memory accesses at \stm{p} are modeled by accessing \stm{mem\_1}, 
memory accesses at \stm{q} are modeled by accessing \stm{mem\_2}.
That way the solver does not need to prove non-interference between the
increment and decrement operations for each spurious counterexample, which
typically results in a dramatic speedup.



%% file: fig/exampleProg.tex
\begin{figure}[t]
  \begin{minipage}{0.47\textwidth}
  \begin{lrbox}{\heapsepExampleLBox}
    \begin{lstlisting}
var mem : [int] : int;
var valid : [int] : bool;

procedure main() {
  var p, q : int;

  call p := malloc();
  call q := malloc();

  mem[p] := 0;
  mem[q] := 0;
  
  while (*) {
    if (*) {
      mem[p] := mem[p] + 1;
    } else {
      mem[q] := mem[q] - 1;
    }
  }
  
  assert mem[p] >= 0;
  assert mem[q] <= 0;
}

procedure malloc() returns (ptr : int);
ensures !old(valid)[ptr];
ensures valid == old(valid)[ptr:=true];
  \end{lstlisting}
  \end{lrbox}
  \resizebox{\textwidth}{!}{
    \usebox\heapsepExampleLBox
  }
  \end{minipage}
  \hfill
  \begin{minipage}{0.47\textwidth}
  \begin{lrbox}{\heapsepExampleRBox}
  \begin{lstlisting}
var mem_1, mem_2 : [int] : int;
var valid : [int] : bool;

procedure main() {
  var p, q : int;

  call p := malloc();
  call q := malloc();

  mem_1[p] := 0;
  mem_2[q] := 0;
  
  while (*) {
    if (*) {
      mem_1[p] := mem_1[p] + 1;
    } else {
      mem_2[q] := mem_2[q] - 1;
    }
  }
  
  assert mem_1[p] >= 0;
  assert mem_2[q] <= 0;
}

procedure malloc() returns (p : int);
ensures !old(valid)[ptr];
ensures valid == old(valid)[p := true];
  \end{lstlisting}
  \end{lrbox}
  \resizebox{\textwidth}{!}{
    \usebox\heapsepExampleRBox
  }

  \end{minipage}
  \caption{Example of a program and its transformation.  The program serves also as the basis of our scalable benchmark
suite. --- The value of
    the variable  \texttt{mem} is a mathematical map.  It is used
    to model the memory.  The program transformation makes the
    independence of the two statements in the loop apparent.
    Intuitively, the two statements use
    the map \texttt{mem}  differently.  The transformation introduces
    \emph{diffent maps for different uses}.
}
  \label{fig_heapsep_example_original_prog}
  \label{fig_heapsep_example_transformed_prog}
\end{figure}



%% file: preliminaries.tex
In this section, we fix our notation regarding program syntax and semantics.

\paragraph{Program Syntax}

We distinguish two types of variables, \emph{map variables} and \emph{base 
variables}.
Map variables are named $\stm{a}, \stm{b}, \ldots$. 
We use $\stm{i}, \stm{j}, \ldots$ for base variables that are used as map
indices in the current context and $\stm{x}, \stm{y}, \ldots$ for all-purpose
base variables.
We use constant (or literal) expressions named $\lit, \lit_1, \lit_2, \ldots$.
We use a special variable $\stm{pc} \in \allvariables$ called
the \emph{program counter}.
We use typewriter font for program variables (e.g., \stm{i}, \stm{x}) and
italics for mathematical variables (e.g., $i$,$x$).

\emph{Expressions} in our programs can have one of three types.
\begin{align*}
  & \text{Expressions of base type:} 
   & \baseexp & \bnf \lit \mid \stm{x} \mid \stmselect{a}{i} \\
  & \text{Expressions of map type:} 
    & \mapexp & \bnf \stm{a} \mid \stmstore{a}{i}{x} \mid \stmconstarray{\lit} \\
  & \text{Boolean expressions:}
    & \boolexp & \bnf \stm{x==y} \mid \prognot \boolexp \mid \boolexp \progand
       \boolexp \mid \boolexp \progor \boolexp 
\end{align*}
The set of all \emph{commands} is generated generated by the following grammar. 
We refer to this set by \allcommands.
\begin{align*}
  c \bnf \stm{x:=\baseexp} \mid \stm{a:=\mapexp}  \mid \stm{havoc x} \mid \stm{havoc a}
        \mid \stm{assume \boolexp}
\end{align*}
The set of \emph{program locations}, \locations, is a set of distinct
identifiers $\{\loc, \loc', \loc_0, \loc_1, \ldots \}$.
A \emph{statement} is a triple of a source program location, a command, and a 
target program location, i.e., 
$\allstatements = \locations
\times \allcommands \times \locations$.
We use the letter $\sigma$ for statements.
Let $\sigma = (\loc, c,\loc')$ be a statement, then we refer to the \emph{source
location} of $\sigma$ by $\stmsrcloc(\sigma)$.
In contexts where the locations are not important we omit them from the
statement and write only the command.
We call statements whose command is of the form \stm{a:=a[i:=x]} \emph{map write
statements}, and we call statements whose command is of the form \stm{x:=a[i]}
\emph{map read statements}.
To highlight that a statement's command is a map write (read), we name the
statement \somewrite (\someread).

A \emph{program} $P$ is given as a control flow graph whose edges are 
statements.
Formally: $P = (\locations, \Sigma, \ell_0)$, 
where 
\locations is a set of locations, 
$\Sigma \subseteq \allstatements$ is a set of statements, and
$\ell_0 \in \locations$ is the initial location.
For technical reasons we do not allow incoming control flow edges at the initial
location.
A program $P$ induces a set of \emph{program variables}, $\variables$, which
are all the variables that occur in any of the statements of $P$. 
We sometimes refer to only the basic variables $\variables_{base} \subseteq
\variables$ or only the map variables $\variables_{map} \subseteq \variables$.
We call the subset of $\Sigma$ that contains all the map write (read)
statements \writestatements (\readstatements).
From now on we assume the program $P$ is given as described here.

We do not allow equating maps in assume statements ($\stm{assume a==b}$).  
In our experience this restriction does not matter in practice.
Furthermore, we only allow equalities between (base) variables, not between
expressions. This is not a proper restriction.

We will abbreviate $\stm{a:=a[i:=x]}$ as $\stm{a[i]:=x}$.
We may omit the case when the store is over a different map, like
$\stm{a:=b[i:=x]}$, from case distinctions, since it can be simulated by
a map update followed by a map assigment; in this case $\stm{a:=b}$ followed by
$\stm{a[i]:=x}$.
Also, we omit chains of stores applied to one map variable; again this omission
does not change the expressiveness of the programming language.


\paragraph{Program Semantics}

For simplicity of presentation we consider only two sorts, namely the base sort
$\unintsort$ and the map sort $\mapsort$.

A \emph{state} in our program is a mapping from program variables to values from
our set of sorts.
The base variables, like $\stm{x}$ and $\stm{i}$ are assigned values of sort
$\unintsort$.
The map variables, like $\stm{a}$, are assigned values of sort $\mapsort$.
The Boolean sort $\{\true, \false\}$ occurs only during evaluation of Boolean
expressions.
The program counter variable \stm{pc} is a special case, its value denotes the location 
$\ell \in \locations$ that the execution is currently in. 

We use the (semantic) map update operator 
$\cdot\subst{\cdot}{\cdot} \colon (\mapsort) \times \unintsort \times \unintsort \to (\mapsort)$: 
Let $a$ be a map, then $a\subst{i}{x}$ is the map
that returns the value $a(j)$ for all arguments $j\neq i$ and the value $x$ for
the argument $i$.

For expressions $\stm{e}$ we give an evaluation function $\eval{\cdot}{\cdot}
\colon \states \times \allexpressions \to (\unintsort \cup (\mapsort))$, which,
given a valuation of the variables, assigns a value to $\stm{e}$:
Every literal has one value in $\unintsort$ it is associated with; the literal
evaluates to that value regardless of state. 
A variable is evaluated by looking up its value in the state.
A map variable's value is a map, a map access at some index evaluates to the
application of the evaluated map value to the evaluated  index value.
The semantics of the store operator is given as the above-mentioned map update
operator.
A constant map expression with some argument $\lit$ evaluates to a map whose
value is $\lit$ at every position.
The Boolean operators are evaluated as usual. Formally:
\begin{align*}
  \eval{s}{\lit} \define & \lit &
  \eval{s}{\stm{v}} \define & s(\stm{v}) \\
  \eval{s}{\stmselect{a}{i}} \define & \eval{s}{\stm{a}}(\eval{s}{\stm{i}}) &
  \eval{s}{\stmstore{a}{i}{x}} \define &
  \eval{s}{\stm{a}}\subst{\eval{s}{\stm{\stm{i}}}}{\eval{s}{\stm{x}}} \\
  \eval{s}{\stmconstarray{$\lit$}} \define & \lambda x \ldotp \lit &
  \eval{s}{\stm{e==e'}} \define & 
    \begin{cases}
      \true & \text{ if } \eval{s}{\stm{e}} = \eval{s}{\stm{e'}} \\
       \false & \text{ otherwise}
    \end{cases}
\end{align*}

The \emph{concrete post} operator $\postOp \colon 2^{\allstates} \times
\allstatements \to 2^{\allstates}$ is given as follows.
\begin{align*}
  \post{S}{\stmedge{\loc}{\stm{x:=\baseexp}}{\loc'}} \define & \{
    s\subst{\stm{pc}}{\loc'}\subst{\stm{x}}{\eval{s}{\stm{\baseexp}}} 
      \mid s \in S, s(\stm{pc}) = \loc \} \\
    \post{S}{\stmedge{\loc}{ \stm{a:=\mapexp}}{ \loc'}} \define & \{
    s\subst{\stm{pc}}{\loc'}\subst{\stm{a}}{\eval{s}{\stm{\mapexp}}} 
      \mid s \in S, s(\stm{pc}) = \loc \} \\
    \post{S}{\stmedge{\loc}{ \stm{havoc x}}{ \loc'}} \define & \{ s\subst{\stm{pc}}{\loc'}\subst{\stm{x}}{v} 
      \mid s \in S, s(\stm{pc}) = \loc, v \in \unintsort \} \\
  \post{S}{\stmedge{\loc}{ \stm{havoc a}}{ \loc'}} \define & \{ s\subst{\stm{pc}}{\loc'}\subst{\stm{a}}{v} 
      \mid s \in S, s(\stm{pc}) = \loc,\\ 
      & \qquad\qquad\qquad\qquad\ \ v \in \mapsort \} \\
  \post{S}{\stmedge{\loc}{ \stm{assume e}}{ \loc'}} \define & \{ s\subst{\stm{pc}}{\loc'} 
      \mid s \in S, s(\stm{pc}) = \loc, \eval{s}{\stm{e}} = \true \} 
\end{align*}


An \emph{execution} $e$ is a sequence of statements and states in alternation, i.e.,
$$e = \execution{s_0}{\sigma_0, \ldots, \sigma_{n-1}, s_n}.$$
Every execution starts in an initial state, i.e., a state $s_0$ where the
program counter \stm{pc} is assigned the initial location \initloc.
Furthermore, the sequence must be consecutive, i.e., for all $i$ from 0 to
$n-1$, the state $s_{i+1}$  must be contained in the set of post states of the
state $s_i$ under the statement $\sigma_i$, i.e., 
$$s_{i+1} \in \post{\{s_i\}}{\sigma_i}.$$  
A special case are the empty executions, an empty execution $s_0$ consists of an
initial state only.
We can write every non-empty execution as $\execution{e}{\sigma, s}$ where $e$ is an
execution.  We denote the set of all executions \allexecutions.

The \emph{reachable states} are all states $s$ such that there is an execution that
ends in $s$.
$$\reach \define \{ s \mid \exists e \in \allexecutions \ldotp e = e' s \}$$

%% file: dataflow.tex
Our program transformation is based on an analysis of the dependencies between
the statements in the program $P$.
In this section, we describe a property that makes explicit which map update 
statements may be reponsible for the value of a map at some index at some
program location. For this, we introduce the relation \lastwrites (read: ``last
writes'') that contains for a potential read in the program all the map updates
that are relevant for that read in some execution of the program.






\paragraph{Last Write Relation \lastwrites}

The relation 
$\lastwrites \subseteq \writestatements \times \readstatements$
relates all map write statements \somewrite
to all the map read statements \someread such that \somewrite is responsible for
the value that is read in \someread in some execution. %



\begin{definition}[Last Writes Relation \lastwrites]
The Last Write relation $\lastwrites \subseteq \writestatements \times \readstatements$ 
  contains a pair $(\somewrite, \someread)$, where the command in \somewrite is of the form
\stmem{a[i]:=x}, and the command in \someread is of the form \stmem{y:=b[j]},
whenever there is an execution $e$ and a value $v$ such that $v$ is written by
  \somewrite and is read by \someread, i.e., if $e$ fulfills the following
  linear time property. 
  \begin{align*}
     \ltldiamond ( \stm{pc} = \stmsrcloc(\somewrite) & \land
    \stm{x} = v \land 
     \ltldiamond (\stm{pc} = \stmsrcloc(\someread) \land \stm{b[j]} = v))
\end{align*}
\end{definition}
In this definition we assume that every value that is written to a map during an
execution is unique; this can be accommodated by providing each value with a
timestamp.
Furthermore, in this definition \stmem{a} and \stmem{b} may or may refer to
the same program variables, the same holds for, \stmem{i} and \stmem{j} and
\stmem{x} and \stmem{y}.

\paragraph{Alternative Characterisation of the Last Writes Relation \lastwrites}
We provide an alternative characterisation of the Last Writes relation
\lastwrites.
This characterisation will lead to an instrumentation of the program that
will allow us to compute an relation \lastwritessharp that overapproximates the
Last Writes relation.


We next define the function  $\lastwrite$ which, given a position $i$,
given a map
\stm{a}, and given an execution $e$, returns the write statement \somewrite that is
responsible for the value that the map \stm{a} has at position $i$ in the last
state of the execution $e$.  
For technical reasons we will use the symbol \nonwrite (to cater for
the case where the map \stm{a} has not been written at
position $i$ in execution $e$).

Formally, we define the function $\lastwrite \colon \mapvars
\times \unintsort \times \allexecutions \to \writestatements \cup
\{\nonwrite\}$
by induction over the length of the execution $e$.
(As explained above, an execution of length 0 is of the form $s_0$ where $s_0$
is an initial state, and an execution of length $n+1$ is of the form
$\execution{e}{\sigma, s}$
 where $\sigma$ is a statement and
$s$ is a state.)
\begin{align*}
  \lastwrite(\stm{a}, j, s_0) \define & 
    \nonwrite \\ 
  \lastwrite(\stm{a} , j, \execution{e}{\stm{havoc a}, s}) \define & 
    \nonwrite \\
  \lastwrite(\stm{a} , j, \execution{e}{ \stm{a:=\stmconstarray{\lit}}, s}) \define & 
    \nonwrite \\
  \lastwrite(\stm{a} , j, \execution{e}{ \stm{a[i]:=x}, s}) \define &
    \begin{cases} 
      \stm{a[i]:=x} & \text{ if } s(\stm{i}) = j \\
      \lastwrite(\stm{a}, j, \execution{e}{}) & \text{ if } s(\stm{i}) \neq j
    \end{cases}\\
  \lastwrite(\stm{a} , j, \execution{e}{ \stm{a:=b}, s}) \define & 
    \lastwrite(\stm{b} , j, \execution{e}{}) \\
  \lastwrite(\stm{a} , j, \execution{e}{ \sigma, s}) \define & 
    \lastwrite(\stm{a} , j, \execution{e}{}) 
    \text{ if  $\execution{e}{ \sigma, s}$ matches none of the above}
\end{align*}

Intuitively, the definition of $\lastwrite(\stm{a}, j, e)$ traces the value of
the map \stm{a} at index $j$ back within the execution $e$ until
it hits the map write statement that is responsible for the fact that \stm{a}
has that value at position $j$ at the end of $e$.
This write statement is returned by \lastwrite.
If the execution consists only of an initial state $s_0$, or the last statement was
a havoc statement with argument \stm{a}, or when \stm{a} has been set to a
constant map by the last statement, then no value in \stm{a} depends on a map write
statement, so \lastwrite returns the symbol \nonwrite.  
If the last statement in the execution has been a write to
map \stm{a}, then \lastwrites checks whether the write was at position $j$. If
that is the case, the last write is returned, otherwise \lastwrite recurses on
the prefix of the execution where the write statement and its successor state
have been dropped.  If the last statement in the execution assigned another map
\stm{b} to \stm{a}, the \lastwrite recurses on the execution prefix, and it
looks for writes on \stm{b} instead of writes on \stm{a}.
Otherwise, the last statement in the execution had no influence on values in
\stm{a}, so it is evaluated recursively on the prefix without the last statement
and state.

As above, the Last Writes relation \lastwrites relates all the write statements
\somewrite to all the read statements \someread, such that there is an execution
where \somewrite is responsible for the value that \someread reads.
From the function $\lastwrite$ we build the explicit characterization of the relation
$\lastwrites \subseteq \writestatements \times \readstatements$
 as follows.
 \begin{align*}
   \lastwrites \define  \{ & (\somewrite, \someread) \mid 
    \someread = (\loc, \stm{x:=a[i]}, \loc')  \\
   & \land \exists \, \execution{e}{ s} \in \allexecutions \ldotp    
       s(\stm{pc}) = \loc \land s(\stm{i}) = i 
       \land \lastwrite(\stm{a}, i, \execution{e}{ s}) = \somewrite \\
     & \land \somewrite \neq \nonwrite \} \\
\end{align*}



%% file: preprocessing.tex
In this section, we present an instrumentation of the program $P$ such that the
Last Writes relation \lastwrites can be expressed in terms of the set of
reachable states of the instrumented program \progpre.


\subsection{Instrumentation}
%
%


We introduce an auxiliary map variable \stm{a-lw} for every map-variable \stm{a} that
occurs in the program $P$.
The values of the maps that are assigned to \stm{a-lw} are not values from our
base sort \unintsort, but instead are symbols that refer to write statements
that occur in $P$.

Intuitively, the transformation is designed in such a way that the fresh
\stm{lw}-maps capture the results of the \lastwrite-function for each program
location.
We construct the transformation in three steps.
We begin by defining by a transformer 
 $\rdtransformer^c \colon \allcommands \to \allcommands$
for some commands whose transformation result does not depend on
their location in the program.

If the command $c$ is a havoc to map variable \stm{a}, or if
$c$ assigns a constant map to \stm{a}, then \stm{a-lw} is assigned a
constant map that contains the symbol \nonwrite at all positions.
This represents that no write statement has an influence on any value in the
map \stm{a} after the command $c$ has been executed.
If $c$ assigns the value of a map variable to another map variable, then
the analogous assignment is done on the respective \stm{lw}-maps.
This expresses that all map write statements that have an influence on \stm{a}
also have an influence on \stm{b} after the command $c$ has been executed.
In all other cases, the transformation $\rdtransformer^c$ leaves the command $c$
unchanged.  
\begin{align*}
  \rdtransformer^c(\stm{havoc a}) \define & 
    \stm{havoc a; a-lw:=\stmconstarray{\nonwrite}} \\
  \rdtransformer^c(\stm{a:=\constarray{\lit}}) \define & 
    \stm{a:=\constarray{\lit}; a-lw:=\stmconstarray{\nonwrite}} \\
  \rdtransformer^c(\stm{b:=a}) \define & \stm{b:=a; b-lw:=a-lw} \\
  \rdtransformer^c(c) \define & c
      \text{ where none of the other cases apply} 
\end{align*} 

From $\rdtransformer^c$ we construct the transformer 
$\rdtransformer^\sigma \colon \allstatements \to \allstatements$,
which transforms the map write statements.
Whenever a map variable $\stm{a}$ is written to at index $\stm{i}$, then
$\stm{a-lw}$ is written at the same index, but with a special value that
identifies the updating statement. 
Statements that are not map write statements are left unchanged by
$\rdtransformer^\sigma$.
\begin{align*} 
  \rdtransformer^\sigma(\somewrite) \define & 
     \stmedge{\loc}{\stm{a[i]:=x; a-lw[i]:=}\somewrite}{\loc'} \\ 
    & \text{ where } \somewrite = \stmedge{\loc}{ \stm{a[i]:=x}}{ \loc'} \\
  \rdtransformer^\sigma((\loc, c, \loc')) \define & (\loc, \rdtransformer(c), \loc')
     \text{ where } (\loc, c, \loc') \not \in \writestatements
\end{align*}

The final statement transformer 
$\rdtransformer \colon \allstatements \to \allstatements$
updates statements that originate from the initial
location \initloc.
Because at the initial location no map writes have been executed, we set every
\stm{lw}-variable to a constant map containing the symbol $\nonwrite$.
(Note that we assume that the initial location has no incoming statements.)
\begin{align*}
  \rdtransformer(\sigma) \define 
    & c \stm{;} \\
    & \stm{a-lw:=\stmconstarray{\nonwrite};} \\
    & \ldots \\
    & \stm{z-lw:=\stmconstarray{\nonwrite}} \\
    & \text{ where } \rdtransformer^\sigma(\sigma) = (\initloc, c, \loc)
     \text{ and }  \mapvars = \{ \stm{a}, \ldots, \stm{z} \}
\end{align*}
%


%


We are now ready to define the instrumented program \progpre.
We define the instrumented program \progpre through applying the transformation
function \rdtransformer to each statement in $\statements$. Formally:
$$\progpre \define \{ \locations, \{ \rdtransformer(\sigma) \mid \sigma \in \Sigma \}, \initloc\}$$

We can now express the Last Write relation \lastwrites through the set of
reachable states of the instrumented program \progpre.
%
\begin{proposition}
  \label{chap_heapsep_prop_lw_vs_a-lw}
  The Last Writes relation \lastwrites as defined in Section
  \ref{chap_heapsep_sec_dataflow} is identical to the relation that relates a
  map write statement \somewrite in \writestatements to a map read statement
  \someread in \readstatements of the form $(\loc, \stm{x:=a[i]}, \loc')$ if
  there is a state $s$ in the set of reachable states of the instrumented
  program \progpre such that the program counter \stm{pc} points to the source
  location of \someread, \loc, and the value that $s$
  assigns to the map read expression \stm{a-lw[i]} is the write statement
  \somewrite.
  Formally:
  \emph{
    \begin{align*}
     \lastwrites = \{ (\somewrite, \someread) & \mid 
       \someread = (\loc, \stm{x:=a[i]}, \loc')  \\
      & \land \exists s \in \reach(\progpre) \ldotp 
	\eval{s}{\stm{pc}} = \loc \land \eval{s}{\stm{a-lw[i]}} = \somewrite\}
    \end{align*}
	}
\end{proposition}

We state the following lemma for later reference (proof of
Theorem~\ref{chap_heapsep_theorem_bisim_P_Pprime} in 
Section~\ref{chap_heapsep_sec_transformation}).

\begin{lemma}
  \label{chap_heapsep_lemma_bisim_P_Ppre}
 $P$ and $\progpre$ are bisimulation-equivalent.
\end{lemma}
%

The proof of this lemma is obvious form the fact that the additional commands
introduced by the transformation is ghost code.

\subsection{Computing an Overapproximation of the Last Writes Relation \lastwrites}

We have seen that the relation \lastwrites can be expressed through the
set of reachable states of the instrumented program \progpre.
The set of reachable states is not computable in general.
Thus, we apply a static analysis that computes an overapproximation of the
set of reachable states.
%

The static analysis must be able to handle programs that manipulate maps.
An example is a static analysis based on the Map Equality
Domain~\cite{Dietsch2018mapequalitydomain}.
This domain is useful to infer equalities and disequalities between expressions
which can involve maps.


We have implemented an extension of the Map Equality Domain.
The extensions supports constraints of the form $\stm{x} \in \{\lit_1, \lit_2\}$ 
which allows us to succinctly express constraints like $\stm{a-lw[i]} \in
\{\sigma_1, \sigma_2\}$. 
Here, $\sigma_1$ and $\sigma_2$ are literals (referring to the corresponding
statements). All literals are pairwise different. 
Thus, these constraints allow us to infer constraints like $\stm{a-lw[i]} \neq
\sigma_3$.
Such constraints are crucial to infer independence of statements.




From now on, we use \lastwritessharp to refer to the overapproximation of the
relation \lastwrites computed by applying the above-described static analysis to
the instrumented program \progpre.  
The static analysis always computes an overapproximation of the set of reachable
states of \progpre. Thus, the relation \lastwritessharp is an overapproximation of
the Last Writes relation \lastwrites.
We state the following remark for later reference (in Lemma
\ref{chap_heapsep_lemma_bisim_Ppre_Pprime}).

\begin{remark}
\label{chap_heapsep_rem_lwsharp_overapp}
 The relation $\lastwritessharp$ is an overapproximation of the Last Write
  relation $\lastwrites$, i.e., 
 $$\lastwritessharp \supseteq \lastwrites.$$
\end{remark}

%% file: transformation.tex
In this section we introduce the program transformation that transforms the
program $P$, given the relation \lastwritessharp, which approximates the Last
Write relation \lastwrites of program $P$.



\subsection{Computing a Partition of the Map Write Statements}

First, we define the relation $R \subseteq \writestatements \times
\writestatements$ that relates all write statement that map
influence the same read statement. 
Two write statements $\somewrite$ and $\somewrite'$ are related by $R$ if there
exists a read statement \someread such that the relation \lastwritessharp
relates both \somewrite to \someread and
$\somewrite'$ to \someread. Formally:
$$R \define \{ (\somewrite, \somewrite') \mid 
\exists \,\someread \in \readstatements \ldotp
\lastwritessharp(\somewrite,\someread) \land 
\lastwritessharp(\somewrite,\someread) \}$$

Based on the relation $R$, we define the relation $r \subseteq \writestatements \times
\writestatements$ as the smallest equivalence relation that contains the
relation $R$.
This equivalence relation $r$ induces a partition over the set \writestatements,
i.e., a set $\writeblocks \subseteq 2^{\writestatements}$ of subsets of the set
\writestatements such that the disjoint union of the subsets is identical to the
original set \writestatements.
Thus, the set \writeblocks consists of disjoint subsets $\{\writeblock_1, \ldots,
\writeblock_n\}$ of the set of all write statements \writestatements. 
The partition \writeblocks has the property that for every two blocks
$\writeblock_1$ and $\writeblock_2$ in \writeblocks, we know that 
if we take one write statement \somewrite from $\writeblock_1$ and another write
statement $\somewrite'$ from $\writeblock_2$, then \somewrite and $\somewrite'$
are independent in the sense that they never have an influence on the same read
statement.

For technical reasons, we add a the singleton consisting only of the symbol \nonwrite  to
\writeblocks. Its use will become clear in the next subsection.

\subsection{Program Transformation}



We introduce a map variable \stm{a\_\writeblock} for each $\writeblock \in
\writeblocks$. 
If for example the write statements \stm{a[i]:=x} and \stm{a[j]:=y} appear in
different blocks $\writeblock_1$ and $\writeblock_2$, then we will replace the
map variable \stm{a} with two different variables $\stm{a\_{$\writeblock_1$}}$ and
$\stm{a\_{$\writeblock_2$}}$ in these statements accordingly.
(There is a subtle point here regarding the fact that \writeblock is a
mathematical object while a variable name consists of characters which we
neglect here.)

We use the notation $\preimage{\lastwritessharp}{\someread}$ to denote the
preimage of \lastwritessharp with respect to some read statement $\someread\in
\readstatements$, i.e.,  %
$$\preimage{\lastwritessharp}{\someread} \define \{ \somewrite
\mid (\somewrite, \someread) \in \lastwritessharp \}.$$


The transformation updates the statements of program $P$ using the 
transformation $\stmtransformer\colon \allstatements \to \allstatements$ as
described in the following.  
The transformation result $\stmtransformer(\sigma)$ depends on the statement
type of $\sigma$.
If $\sigma$ writes to map variable \stm{a}, it is transformed to 
a statement that does the same update to map variable \stm{a\_\writeblock},
i.e., to the map variable corresponding to the block in the partition
$\writeblock \in \writeblocks$ that contains $\sigma$.
If $\sigma$ reads from a map variable \stm{a}, there are two cases. 
Either \lastwritessharp at the read location yields the empty set. 
This means that it is guaranteed that the read position has never
been written to in any execution that reaches $\sigma$. 
In this case, $\sigma$ is transformed to a read from the map variable
\stm{a\_{$\{\nonwrite\}$}} instead of \stm{a}.
Otherwise, by construction of the partition \writeblocks, \lastwrites must yield 
a set that falls completely into a block \writeblock in the partition $\writeblocks$.
In that case, $\sigma$ is transformed to a read from the map variable
\stm{a\_\writeblock} instead of \stm{a}.
If $\sigma$ assigns a map variable \stm{a} to a map variable \stm{b}, it is
transformed to a series of assignments that assign for each block in the
partition $\writeblock \in \writeblocks$ the variable \stm{a\_\writeblock} to
the variable \stm{b\_\writeblock}.
A havoc to a map variable \stm{a} is translated to havoc on all variables
$\stm{a\_\writeblock}$ for every block \writeblock in the partition
\writeblocks, followed by an assume statement that ensures
that all maps \stm{a\_\writeblock} have been set to the same value.
In all other cases, the transformation leaves $\sigma$ unchanged.
Formally:
\begin{align*}
  \stmtransformer((\loc, \stm{a[i]:=x}, \loc')) \define & (\loc, \stm{a\_$W$[i]:=x}, \loc') 
    \text{ where } (\loc, \stm{a[i]:=x}, \loc') \in \writeblock \\
  \stmtransformer(\stmedge{\loc}{\stm{x:=a[i]}}{\loc'}) \define &
       \stmedge{\loc}{\stm{x:=a\_$\{\nonwrite\}$[i]}}{\loc'} \\
       & \text{ if } 
	\preimage{\lastwritessharp}{\stmedge{\loc}{\stm{x:=a[i]}}{\loc'}} = \emptyset\\
  \stmtransformer(\stmedge{\loc}{\stm{x:=a[i]}}{\loc'}) \define & 
     \stmedge{\loc}{\stm{x:=a\_$W$[i]}}{\loc'} \\
      & \text{ if } 
	\preimage{\lastwritessharp}{\stmedge{\loc}{\stm{x:=a[i]}}{\loc'}} \neq \emptyset\\
     &  \text{ and }  
	\preimage{\lastwritessharp}{\stmedge{\loc}{\stm{x:=a[i]}}{\loc'}}
	\subseteq \writeblock\\
  \stmtransformer(\stmedge{\loc}{\stm{b:=a}}{\loc'}) \define & \stmedge{\loc}{\stm{b\_$W_1$:=a\_$W_1$;
  ...; b\_$W_n$:=a\_$W_n$}}{\loc'}  \\
     & \text{ where } \writeblocks = \{\writeblock_1, \ldots, \writeblock_n \} \\ 
  \stmtransformer(\sigma) \define & \sigma \text{ if } \sigma \text{ matches none of the above cases}
\end{align*}

We construct the transformed program $P'$ by replacing all statements $\sigma$
in $P$ by their transformed version $\stmtransformer(\sigma)$.
Formally:
 $$P' \define (\locations, \{ \stmtransformer(\sigma) \mid \sigma \in \Sigma \},
 \initloc)$$

\subsection{Correctness of the Transformation}
In this subsection, we show that the transformation is correct, i.e., that the
program $P$ and the transformed program $P'$ are bisimulation-equivalent.
Given Lemmma~\ref{chap_heapsep_lemma_bisim_P_Ppre}, it is sufficient to prove
the following Lemma.

As an aside: it does not seem obvious to us how to give a bisimulation between
the programs $P$ and $P'$ directly.


\begin{lemma} 
  \label{chap_heapsep_lemma_bisim_Ppre_Pprime}
  The programs $\progpre$ and $P'$ are bisimulation-equivalent.
%
%
\end{lemma}

\begin{proof}
\input{bisimulationproof}
\end{proof}

\begin{theorem}[Bisimulation]
  \label{chap_heapsep_theorem_bisim_P_Pprime}
  $P$ and $P'$ are bisimulation-equivalent.
\end{theorem}
\begin{proof}
  This follows by transitivity of bisimulation-equivalence from Lemmas
  \ref{chap_heapsep_lemma_bisim_P_Ppre}
  and
  \ref{chap_heapsep_lemma_bisim_Ppre_Pprime}. \qed
\end{proof}

%% file: bisimulationproof.tex
We define a bisimulation relation $\bisim$ between $\progpre$ and $P'$ as
follows.

\smallskip
The states $s \in \states_P$ and $t \in \states_{P'}$ are bisimilar, i.e.,
$s\bisim t$, iff
\begin{align*}
   & \forall x \in \variables_{base} \ldotp & \eval{s}{x} & = \eval{t}{x} & (1) \\
    & \text{ and } \\
    & \forall \stm{a} \in \variables_{map} \ldotp \forall \stm{i} \in \variables_{base} \ldotp \\
      & \qquad  (\eval{s}{\stm{a-lw[i]}} = \nonwrite \implies 
      & \kern-2em \forall \writeblock \in \writeblocks \ldotp  
	 \eval{s}{\stm{a[i]}} & = \eval{t}{\stm{a\_$W$[i]}}) & (2a) \\
     & \quad \land (\exists \writeblock  \in \writeblocks \ldotp 
        \eval{s}{\stm{a-lw[i]}} \in \writeblock \implies 
      & \eval{s}{\stm{a[i]}} & = \eval{t}{\stm{a\_$W$[i]}}) & (2b)
\end{align*}

We show, that $\bisim$ is a bisimulation.
Pick $s, t$ such that $s \bisim t$ (We call this the induction hypothesis, I.H.).
Pick $\sigma$ in $\statements_P$ (which corresponds to picking
$\rdtransformer(\sigma)$ and $\stmtransformer(\sigma)$ as well).

We make a case distinction on which statement type $\sigma$ falls into.

\paragraph{Case $\sigma$ is an assignment:}
Let $\{s'\} \in \post{\{s\}}{\rdtransformer(\sigma)}$ and 
let $\{t'\} = \post{\{t\}}{\stmtransformer(\sigma)}$.

First, we consider the conditions (1), (2a), and (2b) with respect to variables
\stm{x}, \stm{a}, and \stm{i} that are not updated by $\sigma$ when $\sigma$ is
deterministic.  
For all three conditions, the reasoning is simple:
By I.H. the condition holds with respect to $s$ and $t$.
Neither $\rdtransformer(\sigma)$ nor $\stmtransformer(\sigma)$ modify \stm{x},
\stm{a} or \stm{i} as they occur in the conditions, and $\rdtransformer(\sigma)$
does not modify \stm{a-lw}. Thus the conditions directly carry over from $s$
and $t$ to $s'$ and $t'$.

In order to prove the conditions for $s'$ and $t'$ with respect to to variables
that are updated by $\sigma$, we make a further case distinction on which type of
assignment $\sigma$ is (and analogously by $\rdtransformer(\sigma)$ and
$\stmtransformer(\sigma)$).
\begin{itemize}
  \item Case $\sigma = \stm{a[i]:=x}$:
    \stm{a} is updated only at position $\eval{s'}{\stm{i}}$; for the other
    positions, the same reasoning as above is applicable.
    We know $\eval{s'}{\stm{a-lw[i]}} = \sigma$ and $\sigma \neq \nonwrite$.
    Thus the antecedent of condition (2a) cannot be fulfilled in $s'$.
    Let $\writeblock \in \writeblocks$ be the block that contains $\sigma$.
    Remember $\stmtransformer(\sigma) = \stm{a\_$W$[i]:=x}$ and
    $\rdtransformer(\sigma) = \stm{a[i]:=x; ...}$. 
    Thus $\eval{t'}{\stm{a\_$W$[i]}} = \eval{t}{\stm{x}} = \eval{s}{\stm{x}} =
    \eval{s'}{\stm{a[i]}}$, which means condition (2b) is fulfilled. 
  \item Case $\sigma = \stm{x:=a[i]}$: Then $\stmtransformer(\sigma) =
    \stm{x:=a\_$\writeblock$[i]}$ for some $\writeblock \in \writeblocks$.
    In order to show $\eval{s'}{\stm{x}} = \eval{t'}{\stm{x}}$, we need to show 
    $\eval{s}{\stm{a[i]}} = \eval{t}{\stm{a\_$\writeblock$[i]}}$.

    First, if $\writeblock = \{\nonwrite\}$, 
    by construction of $\stmtransformer$, we have $\lastwritessharp(\stm{a},
    \stm{i}, s(\stm{pc})) = \{\nonwrite\}$. 
    Thus, by Proposition \ref{chap_heapsep_prop_lw_vs_a-lw} and Remark
    \ref{chap_heapsep_rem_lwsharp_overapp}, we have
    $\eval{s}{\stm{a-lw[i]}} = \nonwrite$. 
    Thus, by condition (2a) in I.H. we get 
    $\eval{s}{\stm{a[i]}} = \eval{t}{\stm{a\_$W$[i]}}$. 
    
    Second, if $\writeblock \neq \{\nonwrite\}$,
    by construction of $\stmtransformer$, we have $\lastwritessharp(\stm{a},
    \stm{i}, s(\stm{pc})) \setminus \{\nonwrite\} \subseteq \writeblock$. 
    Thus, by Proposition \ref{chap_heapsep_prop_lw_vs_a-lw} and Remark
    \ref{chap_heapsep_rem_lwsharp_overapp}, we have
    $\eval{s}{\stm{a-lw[i]}} \in \writeblock$. 
    Thus, by condition (2b) in I.H. we get 
    $\eval{s}{\stm{a[i]}} = \eval{t}{\stm{a\_$W$[i]}}$. 
  \item Case $\sigma = \stm{b:=a}$:
    We must show conditions (2a) and (2b) holds for $s'$ and $t'$ for variable
    \stm{b} and \stm{b-lw}.
    We already showed this for \stm{a} and \stm{a-lw} above (because \stm{a} is
    updated by $\rdtransformer(\sigma)$/$\stmtransformer(\sigma)$).
    Our proof goal follows directly from the fact that $s'(\stm{a}) =
    s'(\stm{b})$ and $s'(\stm{a-lw}) = s'(\stm{b-lw})$ and for all
    $\writeblock \in \writeblocks$, $t'(\stm{a\_\writeblock}) =
    t'(\stm{b\_\writeblock})$ hold, which is ensured by the assignments in the
    statements $\rdtransformer(\sigma)$ and $\stmtransformer(\sigma)$.
  \item Case  $\sigma = \stm{x:=e}$: where $e$ is not a map read.
    Then, we know
    $\rdtransformer(\sigma) = \stmtransformer(\sigma) = \sigma$.
    By I.H., condition (1), $\eval{s'}{\stm{e}} = \eval{t'}{\stm{e}}$ holds,
    because \stm{e} is a base variable or a literal.
    Our goal $\eval{s'}{\stm{x}} = \eval{t'}{\stm{x}}$ follows directly.
\end{itemize} %

\paragraph{Case $\sigma$ is a havoc statement:}
\begin{itemize}
  \item Case $\sigma = \stm{havoc a}$:
    We show show the simulation directions separately.

    First let $s' \in \post{s'}{\rdtransformer(\sigma)}$.
    We need to show existence of an appropriate $t' \in \post{t'}{\stmtransformer(\sigma)}$.
    Given $\eval{s'}{\stm{a}}$, pick $\eval{t'}{\stm{a\_W}}$ for all $W$s
    identical to that. (Clearly, this state $t'$ is not blocked by the assume
    statement in $\stmtransformer(\sigma)$.)

    For the other simulation direction let $t' \in \post{t'}{\stmtransformer(\sigma)}$.
    We need to show existence of an appropriate $s' \in \post{s'}{\rdtransformer(\sigma)}$.
    We know that for all $\writeblock, \writeblock'$, 
    $\eval{t'}{a_\writeblock} = \eval{t'}{a_{\writeblock'}}$ 
    holds (ensured by the assume statement in
    $\stmtransformer(\sigma)$).
    Pick $\eval{s'}{\stm{a}}$ such that it equals all the $\eval{t'}{\stm{a\_W}}$.

  \item Case $\sigma = \stm{havoc x}$:
    We can clearly choose the appropriate $s'$ or $t'$ such that condition (1) is met.
\end{itemize} %

\paragraph{Case $\sigma = \stmem{assume \boolexp}$:}
   Remember we did not allow the use of map variables in assume statements, so
   $\rdtransformer(\sigma) = \stmtransformer(\sigma) = \sigma$.
    Because of I.H., condition (1), $s$ and $t$ agree on all base variables.
    Thus $\eval{s}{\boolexp} = \eval{t}{\boolexp}$.
    Thus whenever an $s'$ is not blocked by $\rdtransformer(\sigma)$, it is not
    blocked by $\stmtransformer(\sigma)$ and vice versa.
\qed

%% file: extensions.tex
The purpose of this paper is to provide formal foundations of a program
transformation that makes independence of groups of map accessing statements
explicit and to prove it correct.
However, we find it important to explant that the approach extends to a full
fledged intermediate language.


We implemented our program transformation in the \ultimate program analysis 
framework\footnote{\url{https://github.com/ultimate-pa/ultimate}}. 
The intermediate representation we support is the most expressive one used by
\ultimate, namely the so-called \emph{interprocedural control flow graph}
(short: ICFG).
ICFGs are control flow graphs whose edges are labeled with 
ith transition formulas. Transition formulas are arbitrary logical
formulas over some background theory that contain an in- and an out-version for
each program variable.
Furthermore, ICFGs allow dedicated edges for procedure calls and returns.
In the following we highlight the most important features that the programming
language used so far does not have and explain what is necessary to support
them.

\paragraph{Multidimensional Maps}
In order to support maps of higher dimensions, we need to slightly adapt the
relation \lastwrites and the corresponding analysis.
On a technical level this is done by having not one but several \stm{lw}-maps
 for each map variable in the original program.
For an $n$-dimensional map variable \stm{a} we would introduce $n$ \stm{lw}-maps
\stm{a-lw-1} to \stm{a-lw-n} where \stm{a-lw-1} is one-dimensional
\stm{a-lw-2} is two-dimensional and so forth.

\paragraph{Transition Formulas}
In transition formulas, the distinction between assume
statements and assignments is not immediately apparent.
For example, given a program variable \stm{a}, the transition formula 
$\stm{a'} = 1$ would correspond to the assignment \stm{a:=1}, while the
transition formula $\stm{a}=1 \land a'=a$ would correspond to the assume
statement \stm{assume a==1}.
In order to infer, how our instrumentation needs to be done, we need to compute,
which which variables are unconstrained in a given formula. Those have to be 
treated like variables subject to a havoc statements are treated.

\paragraph{Procedures}
In order to support procedures, two features are relevant: 
Map-valued parameters must be passed between procedures, and it must be possible
to compute procedure summaries that describe the effect of a procedure on
global map variables (in fact having one of these features would be enough in
terms of expressiveness, but \ultimate supports both). 
Both of these features are enabled by our support for (by-value) assignments
between maps.


%

%




%% file: evaluation.tex

The thorough experimentation needed to establish whether the approach can be
made applicable to classes of practical benchmarks (or, to what
classes) is not in the scope of this paper. In this section, we will
only investigate whether the approach is applicable in principle.
That is, we will use a benchmark
suite which is specifically tailored to condensate the case split
explosion problem.  This helps us to factor out all aspects in
automatic program verification that are orthogonal
to our problem.

We obtain the bechmark suite by starting with the example program from
Section~\ref{chap_heapsep_sec_example}.  The example program
manipulates the map variable \emph{mem} on the two \emph{index} variables \texttt{p}
and \texttt{q}.  We obtain a new program by adding another two
variables and adding the corresponding statements which manipulate
the map variable \emph{mem} on two new variables in the same way as
the existing statements do for \texttt{p}
resp.\ \texttt{q}.   We can iterate the process and thus obtain a
scalable benchmark suite whose  programs have $2$, $4$, $6$, \ldots\
index variables.




\paragraph{Setup}

We ran our experiments on a dedicated benchmarking system, each benchmark task
was limited to 2 CPU cores at 2.4GHz and 20 Gigabytes of RAM.
We ran two toolchains and took three measurements.
One toolchain, called ``\automizer without'', is the standard verification toolchain of the
program verifier \ultimate Automizer.
The toolchain computes an ICFG from the input program and then run's Automizer's
verification algorithm on the ICFG.
  The second toolchain, called
``\automizer with'',
applies our transformation after computing the interprocedural
control flow graph and before running Automizer's verification algorithm.
A third kind of measurements, denoted 
``\automizer after'',
are the timings
of only the verification algorithm  in the toolchain ``\automizer with'', i.e.,
how long the verification of the transformed program takes.

\paragraph{Results}
\input{fig/plot}

In Figure \ref{chap_heapsep_fig_plot} we display the results of our experimental
evaluation.
The x-axis of the plot represents the different example programs, identified by 
the number of map index variables. 
The y-axis represents the time taken by each toolchain.
We ran three toolchains:
The \ultimate Automizer program verifier, \ultimate Automizer  where
before the verification run, the transformation is applied, and a toolchain where
Automizer was run on the already transformed programs.

We observe that the timings of Automizer on the transformed programs are
nearly constant in the number of used map index variables -- the timings range
from 0.9 seconds to 8.8 seconds. This means that the only real difficulty in our
programs lies in deriving the non-interferences between the map accesses.
Furthermore, we can see that the Automizer fails to scale well when it needs to
derive the non-interferences itself: It fails to prove all examples with 10 or
more map index variables.
The toolchain that includes our transformation shows a significantly improved
scaling behaviour even though the transformation (in particular the static
analysis it is based on) is not cheap.

%% file: fig/plot.tex
\begin{figure}[t]
\begin{tikzpicture}
  \begin{axis}[legend pos= north west, xlabel={\# index variables in program},
    ylabel={runtime (s)}, width=\textwidth, height=5cm]
    \addplot[
        scatter,only marks,scatter src=explicit symbolic,
        scatter/classes={
            a={mark=x,red,thick},
            b={mark=o,blue,thick},
	    c={mark=square,violet,thick},
	    d={mark=text,red,thick,}
        },
	text mark={\scriptsize\textbf (TO)}
    ]
    table[x=x,y=y,meta=label]{
        x    y    label

%
%

 02     5.44   a
 04    26.9    a
 06   103      a
 08   412      a

 10  1800      d
 12  1800      d
 14  1800      d
 16  1800      d
 18  1800      d
 20  1800      d

 02     4.7    b
 04     6.92   b
 06    12.3    b
 08    21.5    b
 10    38.0    b
 12    70.6    b
 14   145      b
 16   233      b
 18   390      b
 20   677      b

 02     0.93   c
 04     1.69   c
 06     2.24   c
 08     4.0    c
 10     3.30   c
 12     4.38   c
 14     5.74   c
 16     6.1    c
 18     7.28   c
 20     8.75   c
    };
   \legend{\hspace*{0.2cm}\automizer without, \automizer with \hspace*{0.2cm}, \automizer after\hspace*{0.2cm}}
  \end{axis}
\end{tikzpicture}
  \caption{The \ultimate Automizer toolchain \emph{without} and \emph{with} 
     the program transformation as a preprocessing step,
    and the \ultimate Automizer toolchain in isolation applied \emph{after}
    the program transformation, on a 
    benchmark suite  whose  programs are scaled-up versions of the example
    program in Section~2.  The timeout {{\small{(TO)}}} is set to 1800 seconds.
  \label{chap_heapsep_fig_plot}
  }
\end{figure}
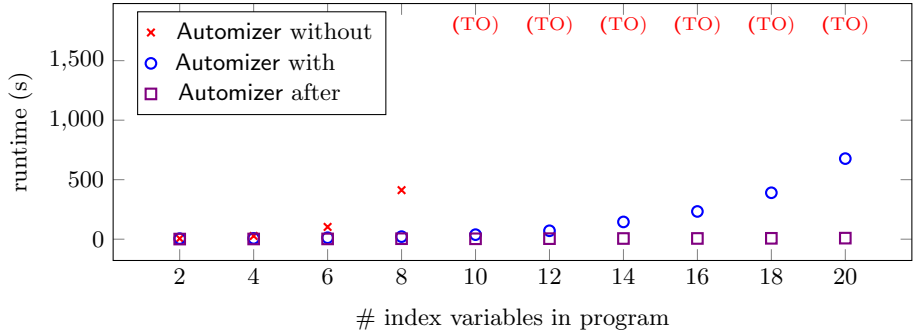

%% file: relatedwork.tex
There are several works resembling ours in that they propose computing
non-interference properties between memory regions to simplify the verification
conditions that are handed to an SMT solver.
Rakamaric and Hu~\cite{DBLP:conf/vmcai/RakamaricH09}, as well as
Wang et al.~\cite{DBLP:conf/vmcai/0062BW17} propose a memory model that
uses maps which are separated according to the results of an up-front alias
analysis.  
Gurfinkel and Navas~\cite{DBLP:conf/sas/GurfinkelN17} propose a related but
different memory model. In their setting, the heap state is passed between
procedures through local map variables. They propose a memory model with a
partitioning that is context-sensitive to improve precision.
In contrast to our work, these papers all rely on C semantics for their input
program, so they do not apply to arbitrary map manipulating programs.



Our relation \lastwrites and the corresponding property is reminiscent of a
large field of work that is concerned with inferring guarantees about data
dependencies between program parts in the presence of arrays. 
We can only mention a few papers here, e.g.,
\cite{DBLP:journals/ijpp/Feautrier91,DBLP:conf/lcpc/PughW93,DBLP:journals/toplas/PaekHP02}.
These papers propose various approaches of finding data dependencies in programs
with arrays in different precisions, for different fragments and for different
applications. None of them is aimed at symbolic program verification as our work
is.
To our knowledge, our property is the only one that accounts for maps, the
crucial difference being the presence of by-value assignments.

%
%

%




%% file: discussion.tex
We discuss some of the choices we made in this paper.

\subsection{Alias Analysis vs. Intermediate Verification Languages}

\begin{figure} 
  \begin{lstlisting}
// (memory model infrastructure)
procedure main() {
  var p, q : int;
  p = 0;
  q = 0;
  // (code not using mem[p] or mem[q])
  p = malloc();
  q = malloc();
  // (code using mem[p] and mem[q])
}
  \end{lstlisting}
  \caption{Program that illustrates why it is not sufficient to
  only consider pointer (map index) variables in our setting.
  Without any additional assumptions we must conclude that \stm{p} and \stm{q}
  may alias and thus that there is a dependency between statements that use 
  \stm{p} and \stm{q} to access the map \stm{mem}. 
  However, if we consider at which program locations \stm{p} and \stm{q} are
  actually used to access the map \stm{mem}, we can conclude that those accesses
  must be mutually independent. (An ensures statement guarantees that the
  procedure \stm{malloc} never returns the same value twice.)
    \label{chap_heapsep_fig_example_ma}
  }
%
\end{figure}

In this subsection, we discuss why classical alias
analyses cannot be used as a basis for our program transformation.

A classical alias analysis reasons about the pointer variables of a program. In
a nutshell, the analysis collects all the assignments in the program that assign
a pointer source value to a pointer variable.
Possible source values are typically: (1) calls to memory-allocating procedures,
like \stm{malloc}, (2) expressions that point to memory that is known to be
implicitly allocated, like the addressof-expression \stm{\&x}, (3) other
pointers.
While the classes of source values may vary, 
it is
always assumed that no two pointers alias ``by accident''. I.e., when a pointer
is uninitialized, it is assumed to be distinct from every other pointer, even
though nothing is known about its value at the time. The same holds for pointers
that have been freed.
Similarly, every pointer that has the value \stm{null} is assumed to not alias
with any other pointer, even if that other pointer also has the value
\stm{null}.
To summarize, only \emph{valid} pointer values are taken into account for alias
analysis. This is sound in the context of the programming language because
accessing an invalid pointer would lead to undefined behaviour according to the
language standard. Thus, the analysis reasons about pointers with the hidden
assumption that no undefined behaviour occurs in the program because in the case
of undefined behaviour all guarantees about what the program does are lost
anyways.

These assumptions enable extremely efficient pointer analyses because in this
setting the only way that two pointers can alias is if there is a chain of
assignments between pointer variables that (transitively) assigns the value of
one pointer variable to the other.
Therefore, a flow-insensitive analysis that collects all assignments of pointer
variables without regard to control flow can already achieve good precision
while being highly scalable.

The analogue to pointers in an intermediate verification language are map
indices, i.e., values that are used to read values from a map variable. It is
common to use mathematical integers as the sort of map indices, like in our
example.
In our setting, assumptions that are not explicitly modelled in the program are
not allowed.
Therefore, we have two options: (1) We model all assumptions in our verification
language. E.g. we would have to check that all pointer accesses are indeed
valid. This is impractical as checking this is a hard verification task on its
own right. (2) We develop an alternative to alias analysis that does not rely on
these assumptions -- which is what we did in this paper.

As the example program in Figure \ref{chap_heapsep_fig_example_ma} illustrates,
it is not enough if our static analysis only considers the values that map
indices may assume.
Instead, we must track when and how (for read or write accesses) the indices are
actually used. This is done by the Last Writes relation \lastwrites.

\subsection{Assume Statements over Map Variables}
 From a theoretical of view, it might be interesting why we omit assume
 statements that equate map variables from our programming language.
 We now explain the complications this would entail.

  Consider the following program snippet.
  \begin{lstlisting}
x := b[i];
a[j]:=y;
assume a==b;
  \end{lstlisting}
  The snippet contains no loops or procedure calls but still the map write in
  the second line influences the map read that comes \emph{earlier} in the code
  because the assume statement establishes a relationship between the maps
  \stm{a} and \stm{b}.
  Thus, because \stm{i} and \stm{j} may alias, we have
  $(\stm{a[j]:=y}, \stm{x:=b[i]}) \in \lastwrites$ (note that the assume
  statement enforces the timestamps to match as well as the values, between
  \stm{a} and \stm{b}).  
  This would mean that a practical computation of \lastwrites would have to
  incorporate both forward- and backward analysis, whereas without such assume
  statements it is sufficient to propagate information in just one
  direction.  

%% file: conclusion.tex
We have investigated the theoretical foundations for a novel research
question which may be relevant for the practical potential of
intermediate verification languages.  The question concerns a
preprocessing step for intermediate verification languages which takes
the similar role that alias analysis plays in the verification for
programming languages.  We have presented a preliminary solution in
the form of a program transformation.  We have integrated the program
transformation into a toolchain.   A preliminary
  experimentation shows that the program transformation
  can be effective, at least in principle.   On a benchmark
suite which is specifically tailored to condensate the case split
explosion problem, the toolchain with the program transformation
scales very well in the size of the program (whereas the toolchain
without the program transformation quickly falls into the case
explosion problem and runs out of time or space).  

The thorough experimentation needed to establish whether the approach can be
made applicable to classes of practical benchmarks (or, to what
classes) is not in the scope of this paper.  We see our investigation as
a preliminary for a wealth of future investigations to explore the practical
potential of intermediate verification languages.